\documentclass[11pt]{article}
\usepackage{ifthen}
\usepackage{booktabs} 
\usepackage{caption}
\DeclareCaptionType{copyrightbox} 
\usepackage{subfig} 
\usepackage{subfloat}
\usepackage{amssymb}

\usepackage{amsmath}
\usepackage{color}
\usepackage{epsfig}
\usepackage{url}
\usepackage{wrapfig}

\newboolean{PODS}

\setboolean{PODS}{true} 
\newcommand{\LIB}[1]{#1}   
\newcommand{\UNLIB}[1]{} 

\newcommand{\PODS}[1]{\ifthenelse{\boolean{PODS}}{#1}{}}
\newcommand{\FULL}[1]{\ifthenelse{\boolean{PODS}}{}{#1}}
\usepackage{amsthm}
\newcommand{\Head}{{\mathit head}}
\newcommand{\Tail}{{\mathit tail}}

\newcommand{\pol}{{\rm pol}}
\newcommand{\Instance}{{\mathit inst}}
\newcommand{\Attr}{{\mathit attr}}
\newcommand{\Label}{\mathit{label}}
\newcommand{\Mark}{\mathit{mark}}
\newcommand{\PD}{\mbox{\it PD}}
\newcommand{\IS}{\mbox{\it IS}}
\newcommand{\FAIL}{{\sc fail}}
\newcommand{\DONE}{{\sc done}}
\newcommand{\NIL}{{\sc nil}}
\newcommand{\CC}{{\cal C}}
\newcommand{\GG}{{\cal G}}

\newcommand{\FC}{\mbox{\sc FC}}

\newcommand{\SUPER}{\mbox{\rm $[\![$LOGSPACE$_\pol]\!]^{\log}$}}

\newcommand{\GC}{\mbox{\rm GC}}

\newcommand{\DSPACE}{\mbox{\sc DSPACE}}
\newcommand{\FDSPACE}{\mbox{\sc FDSPACE}}

\newcommand{\HH}{{\cal H}}

\newcommand{\NNN}{{\mathbb{N}}}

\newcommand{\dspace}[1]{{\mbox{\rm\sc DSPACE[}#1{\rm]}}}

\newcommand{\DUAL}{{\mbox{\sc Dual}}}
\newcommand{\coDUAL}{\overline{\mbox{\sc Dual}}}

\newtheorem{theorem}{Theorem}[section]
\newtheorem{corollary}{Corollary}[section]
\newtheorem{proposition}{Proposition}[section]
\newtheorem{lemma}{Lemma}[section]
\newtheorem{claim}{Claim}[section]
\newtheorem{conj}{Conjecture}[section]
\newtheorem{fact}{Fact}[section]

\newcommand{\nop}[1]{}

\nop{
\newtheorem{thm}{Theorem}

\newtheorem{corollary}[thm]{Corollary}
\newtheorem{proposition}[thm]{Proposition}

\newtheorem{lemma}[thm]{Lemma}

} 

\begin{document}

\FULL{
\conferenceinfo{PODS'13,} {June 22--27, 2013, New York, New York, USA.} 
\CopyrightYear{2013} 
\crdata{978-1-4503-2066-5/13/06} 
\clubpenalty=10000 
\widowpenalty = 10000
}

\vspace{-2cm}

\title{Deciding Monotone Duality and Identifying Frequent Itemsets  in Quadratic Logspace\thanks{This is a preprint of a Paper with the same title, which appeared in: Proceedings of the 32nd ACM SIGMOD-SIGACT-SIGART  Symposium 
               on Principles of Database Systems, PODS 2013, New York, NY, USA, June 22 - 27, 2013, pp. 25--36.}}
\date{}

\FULL{
\author{Georg Gottlob\\
Department of Computer Science and Oxford Man Institute\\
University of Oxford, Oxford OX1\,3QD,  UK\\
georg.gottlob@cs.ox.ac.uk}
\date{}}

\nop{
\numberofauthors{1} 
\author{
\alignauthor
Georg Gottlob \\
       \affaddr{Department of Computer Science and Oxford Man Institute}\\
       \affaddr{University of Oxford}\\
       \affaddr{Oxford OX1\,3QD, UK}\\
       \email{georg.gottlob@cs.ox.ac.uk}}
}
\maketitle

\begin{abstract} 
\vspace{-0.2cm}
The monotone duality problem is defined as follows: Given two monotone formulas $f$ and $g$ in irredundant DNF, decide whether $f$ and $g$ are dual. This problem is the same as duality testing for hypergraphs, that is, checking whether a hypergraph $\cal H$ consists of precisely all minimal transversals of a hypergraph $\GG$. By exploiting a recent problem-decomposition method by Boros and Makino (ICALP 2009), we show that duality testing for hypergraphs, and thus for monotone DNFs,  is feasible in 
$\dspace{\log^2 n}$, i.e., in quadratic logspace. As the monotone duality problem is equivalent to a number of problems in the areas of databases, data mining, and knowledge discovery, the results presented here yield new complexity results for those problems, too. For example, it follows from our results that whenever, for a Boolean-valued relation  (whose attributes represent items), a number of maximal frequent itemsets and a number of minimal infrequent itemsets are known, then it can be decided in quadratic logspace  whether there exist additional frequent or infrequent itemsets. 
\FULL{\vspace{0.5cm}

} 
\end{abstract}

\vspace{-0.3cm}
\nop{
\category{\vspace{-0.2cm}F.2.2}{Analysis of Algorithms and Problem Complexity}{Nonnumerical Algorithms and Problems}[Computations on discrete structures]
\category{\\G.2.2}{Discrete Mathematics}{Graph Theory}[Graph algorithms,hypergraphs]
\category{\\H.2.8}{Database Management}{Database Applications}[Data mining]
\keywords{\vspace{-0.2cm}Duality testing, frequent item set, hypergraph, transversal, data mining.}
\vspace{-0.3cm}
}
\section{{Introduction}}
This paper derives new complexity bounds for the problem {\DUAL} of deciding whether two  irredundant monotone  Boolean formulas in DNF are mutually dual, or, equivalently, of deciding whether two simple hypergraphs are dual, i.e., whether  each of these hypergraphs consists precisely of the minimal transversals of the other.  While the exact complexity remains open, there is progress: We prove \LIB{in the present paper} a DSPACE[$\log^2 n$] upper bound for $\DUAL$, and another,  presumably tighter bound \LIB{for the same problem,} that is expressed in terms of sophisticated machine-bounded complexity classes.  The {\DUAL} problem is actually one of the  most mysterious problems  in theoretical computer science. It has many applications, especially in the database, data mining, and knowledge discovery areas~\cite{EG95,EGJelia,guno-etal-97,hagenDiss}, some of which will be mentioned below. Let us first describe the {\DUAL} problem more formally.

\smallskip

{\bf Duality testing for monotone  DNFs and hypergraphs.}\ 
\LIB{A pair of}\UNLIB{Two} Boolean formulas\LIB{\ $f$ and $g$}\UNLIB{$f(x_1,x_2,\ldots,x_n)$ and $g(x_1,x_2,\ldots,x_n)$} on 
pro\-po\-si\-tio\-nal variables $x_1,x_2,\ldots,x_n$ are {\it dual} if 
$$f(x_1,x_2,\ldots,x_n)\equiv \neg g(\neg x_1,\neg x_2,\ldots,\neg x_n).$$
A monotone DNF is {\em irredundant} if the set of variables in none of its disjuncts is covered by the variable set of any other disjunct.   The {\em duality testing problem} {\DUAL} is the problem of testing whether two irredundant monotone DNFs $f$ and $g$ are dual.

A {\em hypergraph} $\cal H$ is a finite family of finite sets (also called {\em hyperedges}) defined over some set of {\em vertices}  $V({\cal H})$. $\HH$ is simple if no hyperedge is contained in another one.  By default, if $V({\cal H})$ is not explicitly specified, the set of vertices of $\cal H$ is $\bigcup_{E\in {\cal H}} E$. A {\em transversal of} $\cal H$ 
is a subset of $V({\cal H})$  that meets all hyperedges of $\cal H$, and a {\em minimal transversal} of $\cal H$ is a transversal of $\cal H$ that does not contain any other transversal as subset. The set of all minimal transversals of a hypergraph $\cal H$ is denoted by $tr({\cal H})$. 
The {\em Hypergraph Duality Problem} is the problem of deciding for two simple hypergraphs $\cal G$ and $\cal H$ whether ${\cal G}=tr({\cal H})$. Assume $\GG\subseteq tr(\HH)$, 
then, in case $\GG\neq tr(\HH)$, to witness this,  one may want to exhibit a {\em new transversal of $\HH$ with respect to $\GG$}. This is a transversal of $\HH$ that has no hyperedge  of $\GG$ as subset. Obviously, every  new transversal $H$ contains at least one {\em new minimal transversal  of } $\HH$ {\em w.r.t.} $\GG$, but it need not be minimal itself.

It is well known that DNF duality  and hypergraph duality  are actually the same problem (see~\cite{EG95}). In fact, two irredundant monotone DNFs $f$ and $g$ are dual iff their hypergraphs are dual. The hypergraph associated to a monotone DNF has precisely one hyperedge for each disjunct, consisting of the set of all variables of this disjunct. Vice versa, one can trivially associate an irredundant  DNF to each simple hypergraph and thus reduce hypergraph duality to DNF duality. Given that these problems essentially coincide (and can be reduced to each other via trivial reductions that are much easier than logspace reductions), we regard them as one and the same problem, which we refer to as {\DUAL}.

\medskip

{\bf The duality problem in data mining, database theory, and knowledge discovery.}  \  The {\DUAL} problem is at the core of a number of important data mining and database problems. It is central,  for example, to the determination of the maximal frequent and minimal infrequent sets in data mining. More precisely, consider a Boolean-valued data relation $M$ over a set $S$ of attributes called {\em items}, 
and a threshold $z$ with $0<z\leq |M|$. Each subset $U\subseteq S$ is called an {\em itemset}. 
For each tuple $t$ of $M$, let $items(t)=\{A\in S\,\mid\,t[A]=1\}$. The frequency $f(U)$ for an itemset $U$ is the number of tuples $t$ of $M$, such that $U\subseteq items(t)$. U is {\em frequent} if  $f(U)>z$ and infrequent otherwise.

 In data mining, one considers the {\em maximal frequent itemsets}  and the {\em minimal infrequent itemsets}  (under set inclusion) for $M$ and $z$. Let us refer to the former as $\IS^+(M,z)$ and to the latter as  $\IS^-(M,z)$. Clearly, both
$\IS^+(M,z)$ and $\IS^-(M,z)$ are simple hypergraphs over $S$, and we abbreviate them by $\IS^+$ and $\IS^-$, respectively, when $M$ and $z$ are understood. 

The maximal frequent itemsets $\IS^+$ are of great interest in the context of  data mining, but they are hard to compute. In fact, as  shown in~\cite{BGKM2002,BGKM2003} that for a given Boolean-valued relation $M$,  a threshold $z$ and a set $S\subseteq\IS^+(M,z)$, deciding whether there are additional maximal frequent itemsets, i.e., whether $S\neq\IS^+(M,z)$, is NP complete. It follows that, assuming NP$\neq$ P, there cannot be any algorithm for  enumerating $IS^+(M,z)$ with polynomial delay, and that under the slightly weaker assumption NP$\not\subseteq$DTIME[$n^{\small \mbox{polylog}\ \mbox{\em n}}$], there is no algorithm  enumerating $IS^+(M,z)$ with qua\-si\-po\-lynomial delay either. For this reason, rather than computing $\IS^+$ only, one often computes  $\IS^+\cup \IS^-$, which may be exponentially larger in the worst case, but has the advantage of being computable with quasipolynomial delay~\cite{BGKM2003}. As a fundamental result towards the aim of jointly computing $\IS^+$ and $\IS^-$, it was shown 
in~\cite{guno-etal-97} that the minimal infrequent itemsets are exactly the minimal transversals of the complements of the maximal frequent itemsets, i.e. $\IS^-=tr(\,\IS^{+c}\,)$, and thus also $\IS^+=  {tr(\IS^-)}{}^c$, where for $A\subseteq 2^S$,\  ${A^c}=\{S-A|A\in S\}$.
Let {\sc MaxFreq-MinInfreq-Identifi\-ca\-tion} be the following decision problem in data mining: 
Given $M$, $z$, a set  $\GG\subseteq \IS^-(M,z)$, and a set $\HH\subseteq  \IS^+(M,z)$, 
decide whether $\HH= \IS^+(M,z)$ and $\GG=\IS^-(M,z)$,  that is, whether there exists no additional  maximal  frequent or minimal infrequent  itemset for $M$ and $z$, that is not already in $\GG\cup \HH$.  In~\cite{guno-etal-97} it was shown that that there exist no such additional itemset iff 
$\GG=tr(\HH^c)$. 
 With regard to the computational complexity, we thus have:
\begin{proposition}[\cite{guno-etal-97}]
\label{prop:frequent} {\sc \hspace{-2mm}  MaxFreq-MinInfreq-Iden\-ti\-fi\-cation}  is logspace-equivalent to  {\DUAL}. 
\end{proposition}

The results of~\cite{guno-etal-97} are at the base of a host of algorithms for maximal frequent itemset generation, that compute both $\IS^+$ and $\IS^-$ incrementally. These algorithms initialize $\GG$  and ${\HH^c}$ with some easy to compute subsets of $\IS^-$ and $\IS^{+c}$, respectively. 
Then, at each step they check  whether for the current sets $\GG=tr({\HH}^c)$ is true, and if not, compute one or more new transversals from which  new maximal frequent itemsets or minimal infrequent itemsets  can be computed easily, see, e.g.~\cite{ManiToiv,Loo-etal,Guno-etal,BGKM2002,SatohUno}.  Thus, not only the decision problem {\DUAL} is of relevance to data mining, but also the problem of effectively computing a new transversal that acts a witness that $\GG\neq tr({\HH}^c)$. 
In the present paper, we will obtain results on  the complexity of this latter problem, too.

Another interesting related database problem is the {\sc additional key for instance} problem for explicitly given relational in\-stan\-ces. Given a relational instance  $R$ over attribute set $S$, and a set $K$ of minimal keys for $R$, determine if there exists a minimal key for $R$ that is not already contained in $K$.  This problem, which has been shown equivalent to $\coDUAL$ in the early nineties~\cite{EG95}, may be of  renewed interest in the age of Big Data, where massive data tables arise and have to be analyzed, and where the automatic recognition of structural features such as minimal keys may be useful.

\begin{proposition}[\hspace{-0.0mm}\cite{EG95}]
\hspace{-2mm} \,The {\sc additional key for in\-s\-t\-an\-ce}
problem is logspace equivalent to $\coDUAL$. Moreover, enumerating the minimal keys of a relational instance $R$ is equivalent to 
enumerating the set $tr(\HH)$ for some  hypergraph $\HH$ that is logspace-computable from $R$.
\end{proposition}
Other related problems equivalent to $\DUAL$ or to $\coDUAL$ deal with the construction of Armstrong relations for sets of functional dependencies~\cite{EG95}, see also~\cite{gott-libk-90,DemeThi95}.

\smallskip

We also wish to briefly mention a problem from the area of distributed databases. For quorum-based updates~\cite{lamp-78} in distributed databases, the concept of {\em coterie}, which is essentially a hypergraph of intersecting quorums has been introduced, and one is specifically interested in so called {\em non-dominated coteries}  
(for definitions and details, see~\cite{Molina85,IbarakiK93}, and for more recent results and applications, see~\cite{MakinoKamedaPODC,MakinoKameda,HaradaYam}). The following was proven:
\begin{proposition}[\cite{IbarakiK93,EG95}]
A coterie $\HH$ is \LIB{non-domi\-na\-ted} \UNLIB{non-dominated} iff $tr(\HH)=\HH$.  
\end{proposition}  

There are a large number of  applications of the {\DUAL} problem and of hypergraph dualization in the areas of knowledge discovery,  machine learning, and more generally in AI and knowledge representation. Just to mention a few: Learning monotone Boolean CNFs and DNFs with membership  queries~\cite{guno-etal-97}, model-based diagnosis~\cite{ReiterDiag87,grei-smit-wilk-90}, computing a Horn approximation to a non-Horn theory~\cite{kavv-etal-93,gogi-etal-98}, and computing minimal abductive explanations to observations~\cite{Eiter-Makino-Abduction}. Surveys of these and other applications and further references can be found 
in~\cite{EGJelia,EG95,hagenDiss}.  

\smallskip

\LIB{\hspace{-0.16cm}}{\bf Known complexity results}. The   exact complexity of {\DUAL} has remained an open problem.  Fredman and Khachiyan~\cite{fred-khac-96} have shown that {\DUAL} is in  DTIME[$n^{o(\log n)}$], more precisely, that it is contained in DTIME[$n^{4\chi(n)+O(1)}$], where $\chi(n)$ is defined by $\chi(n)^{\chi(n)}=n$. \nop{Note that $\chi(n)\sim\log n/\log\log n = o(\log n)$.}  Eiter, Gottlob, and Makino~\cite{EGM03}, and independently, Kavvadias and Stavropoulos~\cite{kavv-stav-03} have shown that {\DUAL} is in the complexity class co-$\beta_2$P, which means that showing that the complement of {\DUAL} can be solved in polynomial time with $O(\log^2 n)$ nondeterministic bits. This  small amount of  nondeterminism can actually be \UNLIB{improved}\LIB{lowered} to $O(\chi(n)\log n)$ which is $o(\log^2 n)$, see~\cite{EGM03}. 

\smallskip

{\bf Research question tackled} The question about the space-effi\-ciency of {\DUAL}, namely,  whether {\DUAL} can be solved using sub-polynomial or even polylogarithmic space was not satisfactorily answered. It was posed (explicitly or implicitly) several times since  1995, for example in~\cite{EG95,tamaki-00,EGMSurvey}. 
This is the main problem we tackle. In addition, we aim at obtaining a better understanding of the {\DUAL} problem in terms of machine-based structural complexity.

\smallskip

{\bf Results.}
We show \FULL{in this paper}\LIB{in this paper} that \LIB{the decision problem}  {\DUAL} is    in the complexity class 
$\dspace{\log^2 n}$, which is a very low class in {\sc POLYLOGSPACE}.
Modulo the assumption that {\sc PTIME}\LIB{$\not\subseteq$}\UNLIB{is not contained in }{\sc POLYLOGSPACE}, which is widely believed, we thus obtain satisfactory evidence 
that {\DUAL} is not PTIME-hard, which answers another complexity question posed in~\cite{EG95,EGMSurvey}. Our results are based on a careful analysis of a recent problem decomposition method by Boros and Makino~\cite{BorosMakino09}. Their decomposition method actually yields  a parallel algorithm that solves 
{\DUAL} on an EREW PRAM  in $O(\log^2 n)$ time using $n^{\log n}$ processors. However, it is currently not known whether such EREW PRAMS can be simulated 
in $\dspace{\log^2 n}$, and this is actually considered to be rather unlikely. However, Boros' and Makino's algorithm does not seem to exploit the full potential of a PRAM, and by taking into account  the restricted pattern of information flow imposed by the specific self-reductions used in their algorithm, we succeeded to  show \UNLIB{membership  of}\LIB{that} $\DUAL$ \LIB{is} in  $\dspace{\log^2 n}$. 

\hspace{-0.3cm}Complexity theorists have very good reasons to assume that the space class $\dspace{\log^2 n}$ is  incomparable with respect to containment to the 
class co-$\beta_2$P. It is thus somewhat unsatisfactory to have two upper bounds 
for {\DUAL} that are incomparable, which suggests that, most likely, there exist better bounds. This encouraged us to look for a tighter upper bound for {\DUAL} in terms of machine-based  complexity models,  that would be contained in both 
$\dspace{\log^2 n}$ and co-$\beta_2$P, and we succeeded to find one. We can, in fact, show that $\coDUAL$ belongs to the "guess and check" class 
$\GC(\log^2 n, \SUPER)$. This somewhat exotic  new machine-based complexity class \UNLIB{contains precisely}\LIB{consists of}
all problems that can be solved by  first guessing $O(\log^2 n)$ bits 
and then checking the correctness of this guess by a procedure in   
{\SUPER}, which is a complexity class contained in PTIME we will define in the present paper.  We hope that this tighter new bound will provide  a better insight into the very nature of the {\DUAL} problem, and possibly hint at the right direction for future research towards finding a matching upper bound.

\medskip

{\bf Roadmap.} The paper is organized as follows. In the next section we discuss decomposition methods for {\DUAL} and give a succinct description of the method of Boros and Makino, which we consider to be the currently most advanced method. In Section 3, we define complexity classes based on iterated self-compositions of functions and prove a useful complexity-theoretic lem\-ma. In Section 4, we use this lemma to prove our main result, namely that {\DUAL} is in   $\dspace{\log^2 n}$. In section 5 we provide our tighter structural complexity bound for {\DUAL}. The paper is concluded in Section 6, where we also exhibit a diagram (Fig.~1) that puts all relevant complexity classes  
in relation, and highlights the new upper bounds.
\FULL{\section{The Decomposition Method by Boros and Makino}}\PODS{\section{Decomposition Method by 
	Boros and Makino}}
\label{sec:bm}
Most algorithms for deciding {\DUAL} rely on decompositions that start with an original {\DUAL} instance and recursively transform it into a conjunction of smaller instances, until each instance is either seen to be a no-instance because it violates necessary conditions for duality, or until it is small and efficiently decidable. Such decompositions are also known as {\em self-reductions}, see, e.g., Section 5.3 of~\cite{EGM03}.\nop{\footnote{Actually, there, {\em disjunctive self-reductions of the complement of \DUAL} were considered.}.}\nop{Such a self-reduction is applied recursively until all instances have small size and are trivially solvable, or until one of the generated instances does not satisfy some necessary condition for duality that can be checked efficiently.} The decomposition process corresponds  in the obvious way to a {\em decomposition tree}.  Different decomposition methods give rise to decomposition trees of different shapes and depths. For example, the well-known algorithm A by Fredman and Khachiyan~\cite{fred-khac-96} produces a "skinny" binary decomposition tree of depth linear in the input volume $|{\cal G}|\times|{\cal H}|$, while their algorithm B produces a non-binary tree of similar depth, but with fewer nodes. Later, decomposition methods giving rise to  trees of polylogarithmic  depth were published. In particular, the methods of Kavvadias and Stavropoulos~\cite{kavv-stav-03}  as well as the two methods by Elbassioni in~\cite{Elbassioni08} give rise to decomposition trees of 
polylogarithmic depth. Finally, decomposition methods yielding  trees of  logarithmic depth  were presented by Gaur~\cite{gaur-99} (see also Gaur and Krishnamurti~\cite{gaur-krish-00}), and, more recently, by Boros and Makino~\cite{BorosMakino09}. As we will show, the logarithmic-depth decomposition trees
generated by these methods can be used to show that {\DUAL} is in  DSPACE[$\log^2 n$]. In particular, we use the elegant decomposition method of Boros and Makino~\cite{BorosMakino09} to prove this, but we could have used Gaur's method~\cite{gaur-99} in a similar fashion.  In the rest of this section, we give a succinct description of the  method of Boros and Makino,  that contains all the essentials we need for our subsequent complexity analysis. It is assumed that the input instance $I=(\GG, \HH)$ we have  $|\HH|\leq |\GG|$,  and that $\GG\subseteq tr(\HH)$ and $\HH\subseteq tr(\GG)$. Clearly this can be  tested in logarithmic space. 

For an input instance  $I=(\GG, \HH)$ of {\DUAL} over a vertex set $V$, let  $T(\GG,\HH)$ denote its decomposition tree. Let $\aleph_\HH =\NNN^0\cup\NNN^1\cup \NNN^2\cup\cdots\cup\NNN^{\lfloor \log |\HH|\rfloor}$, where $\NNN$ denotes the natural numbers, and where $\NNN^0$ is defined to contain the empty  sequence $()$ only, which has length 0. Thus $\aleph_\HH$ contains precisely all sequences of natural numbers of length up to  $\lfloor \log |\HH|\rfloor$. 

Each node  of $T(\GG,\HH)$ has five data structures associated with it: 
\begin{description}
\item[(i)\phantom{ii}]A unique label $\Label(\alpha)$ consisting of a  sequence 
in $\aleph_\HH$. In particular, the root $\alpha_0$ of
$T(\GG,\HH)$ is labeled by $()$, and the $i$-th child of a node labeled $(j_1,\ldots,j_k)$ 
is labeled $(j_1,\ldots,j_k,i)$.
\item[(ii)\phantom{i}]A set $S_\alpha\subseteq V(\GG)$.
\item[(iii)]An instance of {\DUAL} $\Instance(\alpha)=(\GG^{S_\alpha},\HH_{S_\alpha})$, \\
where  $ \GG^{S_\alpha}=\{E\cap S_\alpha\mid  E\in \GG\}\ \mbox{\rm and}\ 
 \HH_{S_\alpha}=\{E\in\HH\mid E\subseteq S_\alpha\},$.
\item[(iv)]A marking $\Mark(\alpha)\in\{\mbox{\sc \DONE, \FAIL, \NIL \}}$, where each leaf of the final decomposition tree will be marked with {\DONE} or {\FAIL}, and each non-leaf will be marked with dum\-my value {\NIL}.
Intuitively, each leaf marked {\DONE} identifies a branch that does not contradict $\HH=tr(\GG)$, 
whereas a leaf marked {\FAIL} identifies a branch that proves that $\HH\neq  tr(\GG)$.
\item[(v)] A set of vertices  $t(\alpha)\subseteq V(\GG)$. This set will be the empty set for each node not marked {\FAIL}, and, in case $\alpha$ is marked {\FAIL},
will contain a witness for $\HH\neq  tr(\GG)$ in form of a new transversal of $\GG$ with respect to $\HH$.
\end{description}

Let us now describe the method for building $T(\GG,\HH)$ and deciding whether $\HH=tr(\GG)$ in detail. At each stage of the algorithm, let us denote the set of current leaf-nodes by $\Lambda$. Here is how the tree is built. The input instance $(\GG,\HH)$ is first transformed into a  initial tree consisting of the root $\alpha_0$ with $\Label(\alpha_0)=()$, $S_{\alpha_0}=V$, 
$\Instance(\alpha_0)=(\GG,\HH)$, $\Mark(\alpha_0)=$\NIL, and $t(\alpha_0)=\emptyset$.
At each stage of the decomposition, first, each leaf $\alpha\in \Lambda$ where $|\HH_{S_\alpha}|\leq 1$, will be marked by the following procedure, and will then not be further expanded and will thus be a leaf of the final tree $T(\GG,\HH)$:

\bigskip

\hrule
\LIB{\smallskip}\UNLIB{\medskip}
\noindent
{\sc procedure marksmall($\alpha$)}: 
\UNLIB{\begin{description}
\item{\sc case 1.} {\sc if} $|\HH_{S_\alpha}|=0$ and $\emptyset\not\in \GG^{S_\alpha}$, {\sc then} 
$\{\, \Mark(\alpha):=$\,{\FAIL};\ $t(\alpha) := S_\alpha\, \}$.
\item{\sc case 2.} If $|\HH_{S_\alpha}|=0$ and $\emptyset\in \GG^{S_\alpha}$,  {\sc then} $\{\, \Mark(\alpha):=$\,\DONE; $t(\alpha)=\emptyset\,\}$.
\item{\sc case 3.} If $\HH_{S_\alpha}=\{H\}$ and $\{\{i\}| i\in H\}\subseteq \GG^{S_\alpha}$, {\sc then} $\{\, \Mark(\alpha):=$\,\DONE; $t(\alpha)=\emptyset\,\}$.
\item{\sc case 4.} {\sc otherwise}, let $H$ denote the only hyperedge of 
$\HH_{S_\alpha}$, and set $\Mark(\alpha):=$\,{\FAIL}, and  $t(\alpha):= S_\alpha-\{i\}$ for some arbitrarily chosen $i\in H$ with $\{i\} \not\in \GG^{S_a}$.\\
\LIB{\hrule}
 \end{description}}
\LIB{\begin{description}
\item{\sc case 1.} {\sc if} $|\HH_{S_\alpha}|=0$ and $\emptyset\not\in \GG^{S_\alpha}$, {\sc then}\\ 
$\{\, \Mark(\alpha):=$\,{\FAIL};\ $t(\alpha) := S_\alpha\, \}$.
\item{\sc case 2.} If $|\HH_{S_\alpha}|=0$ and $\emptyset\in \GG^{S_\alpha}$,  {\sc then}\\ $\{\, \Mark(\alpha):=$\,\DONE; $t(\alpha)=\emptyset\,\}$.
\item{\sc case 3.} If $\HH_{S_\alpha}=\{H\}$ and $\{\{i\}| i\in H\}\subseteq \GG^{S_\alpha}$, {\sc then} $\{\, \Mark(\alpha):=$\,\DONE; $t(\alpha)=\emptyset\,\}$.
\item{\sc case 4.} {\sc otherwise}, let $H$ denote the only hyperedge of 
$\HH_{S_\alpha}$, and set $\Mark(\alpha):=$\,{\FAIL}, and  $t(\alpha):= S_\alpha-\{i\}$ for some arbitrarily chosen $i\in H$ with $\{i\} \not\in \GG^{S_a}$.\\
\hrule
\end{description}}
\UNLIB{\smallskip}
\UNLIB{\hrule}

\bigskip

Then, each leaf $\alpha$  of $\Lambda$ not yet marked is subjected to the following procedure:

 \bigskip
\hrule

\medskip
\noindent
{\sc procedure process($\alpha$)}: 
\begin{enumerate}
\item Let $I_\alpha$ consist of those vertices of $\HH_{S_\alpha}$ that occur in more than 
$|\HH_{S_\alpha}|/2$ hyperedges of $\HH_{S_\alpha}$;

\item {\sc if} $I_\alpha$ is a new transversal of $\GG^{S_\alpha}$ with respect to $\HH_{S_\alpha}$,
{\sc then} \\ $\{\, \Mark(\alpha):=$\,\FAIL; $t(\alpha):= I_\alpha$;  {\sc exit procedure}\};

\item {\sc otherwise if} there is a $G\in \GG^{S_\alpha}$ such that $G\cap I_\alpha=\emptyset$ {\sc then}
let 
$$\CC= \{S_\alpha-(E-\{i\})| E\in\GG^{S_\alpha}_G\, \mbox{\it and}\ i\in E\cap G\},$$
where  $\GG^{S_\alpha}_G=\GG^{S_\alpha}-\{E'\in  \GG^{S_\alpha}|\ E'\subseteq S_\alpha-G\}$;
\item {\sc otherwise if} there exists a $H\in \HH_{S_\alpha}$ such that $H\subseteq I_\alpha$
 {\sc then} let $$\CC = \{S_\alpha-\{i\} | i\in H \}\cup \{H\};$$
\item Let $\kappa(\alpha)=|\CC|$ and 
\LIB{assume $\CC=\{C_1, C_2,\ldots,C_{\kappa(\alpha)}\}.$}
\UNLIB{the elements of $\CC$ be $C_1, C_2,\ldots,C_{\kappa(\alpha)}$.} For each $C_i$, $1\leq i\leq \kappa(\alpha)$, create a new child  $\alpha_i$ with $\Label(\alpha_i)=(\Label(\alpha), i)$, \UNLIB{\\}
$S_{\alpha_i}=C_i$, $\Instance(\alpha_i)=(\GG^{S_{\alpha_i}},\HH_{S_{\alpha_i}})$, $\Mark(\alpha_i)=$\,\NIL, and 
$t(\alpha_i)=\emptyset$.
\nopagebreak
 \end{enumerate}
\nopagebreak
\hrule
\FULL{\pagebreak}
\medskip

\PODS{\bigskip}

Exhaustively apply the procedures {\sc marksmall} (to unmarked leaves  $\alpha$ having  $|\HH_{S_\alpha}|\leq 1$) and {\sc process} (to all other unmarked leaves), until there are no unmarked leaves left in the tree. The resulting tree is  then $T(\GG,\HH)$. 
\LIB{

}
 Note that, due to the possible multiple choices of 
$i$ in {\sc case 4} of the {\sc marksmall} procedure, and of 
$G$ in Step~3 and of $H$ in Step~4 of the {\sc process} procedure, the decomposition tree $T(\GG,\HH)$ is actually not uniquely defined. However, this is not a problem. To obtain a well-defined decomposition tree $T(\GG,\HH)$, we may resort to any pair of deterministic versions of {\sc marksmall} and of {\sc process}, for example, we may use the version of {\sc marksmall} where in {\sc case 4} the smallest $i\in H$ fulfilling $\{i\} \not\in \GG^{S_a}$ is chosen, and the version of {\sc process}
where 
in Step 3 the lexicographically first edge 
 $G\in \GG^{S_\alpha}$  with $G\cap I_\alpha=\emptyset$ is chosen, and similarly for $H$ in Step~4.

The following proposition  summarizes important results by Boros and Makino~\cite{BorosMakino09}.
\FULL{\smallskip}

\begin{proposition}[Boros and Makino~\cite{BorosMakino09}]\label{prop:bm} \hfill
\begin{enumerate}
\item $\HH =tr(\GG)$ iff all leaves of $T(\GG,\HH)$ are marked \DONE. 
\item The depth of $T(\GG,\HH)$ is bounded by  $\log |\HH|$. 
\item  Each node $\alpha$ of $T(\GG,\HH)$
has at most $|V|\cdot|\GG|$ children, i.e., $\kappa(\alpha)\leq |V|\cdot |\GG|$, where $V$ is the set of vertices of $\GG$ and $\HH$. 
\item If $\HH\neq tr(\GG)$, then $T(\GG,\HH)$ has at least one leaf labeled
{\sc fail}, and the set  $t(\alpha)$ associated to each leaf $\alpha$ labeled {\sc fail} is a  new transversal of $\GG$ \UNLIB{with respect to}\LIB{w.r.t.} $\HH$.
\end{enumerate}
\end{proposition} 

\section{{A Complexity-Theoretic Lemma}}
\label{sec:lemma}

For a space-constructible numerical function $z$,  $\DSPACE[z(n)]$ \, (resp.  $\FDSPACE[z(n)]$),\, denotes, as usual, the class of all all decision problems  (resp. computation problems)  solvable deterministically in $O(z(n))$ space. For a function $f$, let $f^1=f$ and for $i\geq 1$, let $f^{i+1}=f\circ f^i$, where $\circ$ is the usual function composition, i.e., where for each $x$ in the domain of $g$, $(f\circ g)(x)=f(g(x))$. 
Let $Q_{\log}$ denote the set of all functions $\rho$
from strings over some input alphabet to the non-negative natural numbers,  where for each input string $I$, $\rho(I)$ is $O(\log |I|)$ and such that $\rho$ is logspace-computable, i.e., $\rho(I)$ is computable in logarithmic space from $I$.
 For each function $f$, and for each function 
 $\rho\in Q_{\log}$, let  $f^\rho$ 
denote the function that to each input $I$ associate the output $f^\rho(I)=f^{\rho(I)}(I)$.  
If $\FC$ denotes a functional complexity class, then $[\![\FC]\!]^{\log}$ denotes the class of functions that can be built from some function $f$ in $\FC$ via a logarithmic number $\rho(I)=O(\log n)$ of self-compositions of $f$ for each input  $I$ of size $n$:
$$[\![\FC]\!]^{\log}=\bigcup_{f\in\mbox{\footnotesize \sc FC},\, \rho\in \mbox{\small Q$_{\log}$}} \hspace{-0.2cm}\{f^{{\rho}}\}.$$
For a functional complexity class $\FC$, the subclass $\FC_\pol$ is defined as the set of all functions $f$ of $\FC$ for which 
there exists a polynomial $\gamma$ such that for each input $I$, and for each $i\geq 1$, $|f^i(I)|\leq \gamma(|I|)$.
In general,  $\FC_\pol$ is a proper subclass of $\FC$. This is, in particular so 
for  $\FDSPACE[\log n]$, i.e., functional logspace, a.k.a. FLOGSPACE.  To see this, let $f$ be the 
function that associates to an input  of size $n$ an output consisting of $n^2$ zeros. 
Clearly,  $f\in \FDSPACE[\log n]$, 
but the output sizes of the $f^i$ 
are not bounded by any fixed polynomial when $i$ grows. 
Thus, $f\not\in[\![\FDSPACE[\log n]_\pol]\!]$, 
hence $[\![\FDSPACE[\log n]_\pol]\!]$ is a proper subclass of  the class \LIB{\\}$\FDSPACE[\log n]$.

\begin{lemma}
\label{Lemma1}
\LIB{$}$[\![\FDSPACE[\log n]_\pol]\!]^{\log}\subseteq\, \FDSPACE[\log^2n].$\LIB{$}
\end{lemma} 
\begin{proof} 
The proof is similar the well-known proof that for any two functions $f,g$ that are in the functional class $\FDSPACE[\log n]$,
their composition $g\circ f$  is in $\FDSPACE[\log n]$, too. (See, e.g.~\cite{papa-94}). However, here, the logarithmic (rather than constant)  number of compositions is responsible for the blowup of the required space by a logarithmic factor. 
Let $f$ be a function from strings to strings in $\FDSPACE[\log n]_\pol$, realized by a logspace Turing Machine $T$,  and let $\rho\in Q_{\log}$. 
In order to prove the lemma, it is sufficient to show that one can construct a single functional  Turing machine $T^*$
with space bound $O(\log^2 n)$ that simulates the pipelined application $T^{\rho(I)}$ that outputs $f^{\rho(I)}(I)$. \LIB{

}
$T^*$ first computes $\rho(I)$ in logspace and then simulates an arrangement  of $\rho(I)$ copies of $T$, say,  $T_1,T_2,\ldots,T_{\rho(I)}$, such that 
the input string $v_1$ to $T_1$ is $I$, and such that for $i\geq1$, the input string $v_{i+1}$ to $T_{i+1}$ is equal to the output string $w_i$ of $T_i$. Given that the size of $w_i=T^i(I)$ is bounded  by some fixed  polynomial $\gamma$, there are numbers $a$ and $b$ such that each $T_i$ requires no more than space $a+b\log n$. 
When simulating the pipelined computation $T_{\rho(I)}(T_{\rho(I)-1}(\cdots(T_2(T_1(I))))$ on a single Turing machine $T^*$, we have to avoid the effective storage of any intermediate output $w_i$ (or, equivalently, input $v_{i+1})$.
To this aim, $T^*$ simulates each $T_i$ via a logspace procedure $P_i$ that 
maintains its own space area  on the worktape of $T^*$. Each $P_i$
acts like $T_i$, except for the following modifications: For  $1<i<\rho(I)$, $P_i$ has a single output bit which is stored on the worktape of $T^*$;  moreover $P_i$  takes as input a dedicated special index register $d_i$ that specifies which output bit of $T_i$ is to be computed, and computes only this output bit  (suppressing all other output bits)  and stores it in a single-bit register $o_i$. $T_i$'s access to its $j$-th input bit is then simulated by 
$P_i$ writing "$j$" (in binary)  into the special index register $d_{i-1}$, starting $P_{i-1}$, and then waiting until $P_{i-1}$ writes  the desired output bit into $o_{i-1}$ which corresponds to the correct value of the $j$-th output of $T_{i-1}$, and thus the $j$-th input bit to $T_i$. $P_1$ and $P_{\rho(I)}$ work in a similar way, except that $P_1$ directly accesses the  input string $I$ from the input tape of $T^*$, and $P_{\rho(I)}$, rather than suppressing some output bits,  writes {\em all} output bits to the output tape of $T^*$. 

The workspace required by each procedure $P_i$ is easily seen to be bounded by $a' +b'\log n$ for some fixed constants $a'$ and $b'$ independent of $n$. This reflects the $a +b\log n$ bits required to execute $T_i$, plus  the little extra space $P_i$ may require for its index $d_i$, for the output bit $o_i$,  and for a constant number of auxiliary counters and pointers (of size at most $a+ b\log n$  bits each) for control and stack management for the $P_i$ procedures.  Given that $\rho(I)$ is $O(\log n)$, $T^*$ requires $O(\log^2 n)$ space in total.  
\end{proof}

\smallskip

Note that the same space bound doesn't  hold for the com\-ple\-xi\-ty class $[\![\FDSPACE[\log n]]\!]^{\log}$. In fact, with functions $f$ in this class, intermediate outputs $f^i(I)$ may be of superpolynomial size, and in the worst case, even of exponential size  $n^{\Theta(n)}$. Therefore, 
when omitting the ``pol" restriction, 
\UNLIB{the best space bound we are able to show is  $[\![\FDSPACE[\log n]\,]\!]^{\log} \subseteq \mbox{\sc FPSPACE}.$}\LIB{$[\![\FDSPACE[\log n]\,]\!]^{\log} \subseteq \mbox{\sc FPSPACE}$ is the best space upper bound we are able to show.} Since an 
$\FDSPACE[\log^2 n]$ Turing machine has an output of size at most $n^{O(\log n)}$, it \UNLIB{is actually the case}\LIB{it actually holds} that  $[\![\FDSPACE[\log n]]\!]^{\log}$ $ \not\subseteq$  $\FDSPACE[\log^2 n]$.

\FULL{\section{The New Space Bound for \bfseries{\scshape{Dual}}}}
\PODS{{\section{{The New Space Bound}}}}
\label{sec:logsquare}
The main result of this section is  that  for a pair $(\GG,\HH)$, the entire decomposition tree $T(\GG,\HH)$ (with all markings and labels) produced by the decomposition method of Boros and Makino as outlined in Section~\ref{sec:bm} can be computed with quadratic logspace. The other space-complexity results follow from this as simple corollaries.

We start with a lemma that provides us with a logarithmic space bound for computing the $i$-th child of a node $\alpha$ of the decomposition tree from the fully labeled node $\alpha$ and from the set $V$ of vertices of the original input instance, or for discovering that such a child does not exist. 
If $\alpha$ is a node of the decomposition tree, let us denote by $\Attr(\alpha)$ the attributes of $\alpha$, i.e., the tuple $(\Label(\alpha),S_\alpha, \Instance(\alpha),\Mark(\alpha),t(\alpha))$. 

\begin{lemma} 
\label{lem:next} 
There \FULL{is} \PODS{exists} a deterministic logspace  procedure \UNLIB{\PODS{\\}} {\sc next}$(V,\Attr(\alpha),i)$, which for each {\DUAL}  instance  $(\GG,\HH)$ over vertex set $V$, for each 
attribute set $\Attr(\alpha)$ of a node 
$\alpha$ of $T(\GG,\HH)$, 
and for each positive integer $i\leq |V|\cdot|\GG|$ outputs: 
\begin{itemize}
\item $\Attr(\alpha_i)$ if  $\alpha_i$ is the $i$-th child of $\alpha$ in $T(\GG,\HH)$;
\item {\sc impossible} otherwise (i.e., if $\alpha$ has less than $i$ children).
\end{itemize}  
\end{lemma} 
\begin{proof}
First note that by simple inspection it is immediate that the procedures {\sc marksmall} and {\sc  process} given in Section~\ref{sec:bm} can be implemented by deterministic logspace transducers. In fact, these procedures only perform 
a fixed composition of 
simple cardinality checks, counting, assignments, and set theoretic operations that are all well-known to run in logspace. 

A procedure {\sc next}, as required, can be constructed as follows. If $\Label(\alpha)\in \{\mbox{\DONE, \FAIL}\}$ then output {\sc impossible},  \UNLIB{otherwise}\LIB{else} perform $\mbox{\sc marksmall}\!^*(\mbox{\sc process}^*(\alpha))$, where:  
\LIB{\vspace{-3mm}}
\begin{itemize}
\item {\sc process$^*$} works like {\sc process} except that it outputs only the $i$-th child of $\alpha$, if such a child exists,  rather than outputting all children, and output {\sc impossible} otherwise; and 
\item {\sc marksmall\!$^*$} works like {\sc marksmall}, except that it also accepts the input 
{\sc impossible}, in which case it also outputs {\sc impossible}.
\end{itemize}
These minor modifications of {\sc marksmall} and {\sc  process} clearly run in deterministic logspace, therefore \UNLIB{also} their com\-position does, and hence so does the procedure {\sc next}.\end{proof}

\medskip

A {\em path descriptor} for a {\DUAL} instance $I=(\GG,\HH)$ over a vertex set $V$ is a list of length $\leq \lfloor\log |\HH|\rfloor$, whose  elements are  integers bounded by $|V|\cdot|\GG|$. 
The set of all path descriptors for $I$ is denoted by $\PD(I)$. Clearly, $\PD(I)\subset\aleph_\HH$, and each label $\Label(\alpha)$ of a node $\alpha$ of $T(\GG,\HH)$ is contained in $\PD(I)$. Intuitively, a path descriptor, exactly in the same way as a label, is intended to describe a sequence of child-indices, that, starting from the roof of $T(\GG,\HH)$ lead to a specific node $\alpha$ of $T(\GG,\HH)$. The root of $T(\GG,\HH)$ is identified by the empty path descriptor. If $\pi=(i_1,i_2,\ldots,i_r)$ is a path descriptor, then $\Head(\pi)=i_1$ and $\Tail(\pi)$ is the path descriptor $(i_2,\ldots,i_r)$. Two path descriptors of the form 
$(i_1,\ldots,i_r)$ and $(i_1,\ldots,i_r,i_{r+1})$ are said to be {\em consecutive}. 

\smallskip
\begin{lemma}
\label{lem:pathnode}{}
\LIB{\phantom{x}\hspace{3mm}}There is a procedure {\sc pathnode$(I,\pi)$} that \LIB{\\}runs in deterministic space $O(\log^2(|I|))$, 
that for each {\DUAL} input instance $I$ and path descriptor $\pi\in \PD(I)$ outputs
$\Attr(\alpha)$ if $\pi$ corresponds to the label $\Label(\alpha)$ of a node $\alpha$ in $T(\GG,\HH)$,
and outputs {\sc wrongpath} otherwise.
\end{lemma}
 
\begin{proof}
Let $I=(\GG,\HH)$, $V=V(\GG)$, and $\pi\in\PD(I)$, and let $\ell(\pi)$ denote the length of the sequence $\pi$  (recall that ($\ell(\pi)\leq \log |I|$). The procedure {\sc pathnode} first computes in deterministic  logspace $\Attr(\alpha_0)$ for the root $\alpha_0$  of $T(\GG,\HH)$. It then computes 
$f^{\ell(\pi)}(V, attr(\alpha_0), \pi)$, 
where $f$ is the function corresponding to the procedure $F$ described as follows. $F$ accepts as input either the string {\sc wrongpath}, or a triple $(W, \Attr, \gamma)$ where $W$ is a set, $\Attr$ is a data structure of the same format as the attributes $\Attr(\beta)$ of some vertex $\beta$ in a decomposition tree, and $\gamma$ is  a sequence of positive integers.  On all other inputs, $F$  outputs the empty string.
On input {\sc wrongpath}, 
$F$ outputs {\sc wrongpath}; otherwise $F$ computes $F'(\mbox{\sc next}(W, \Attr, head(\gamma))$, where {\sc next} be as specified in Lemma~\ref{lem:next}, and where $F'$ is as follows. $F'$ outputs {\sc wrongpath} if  
$\mbox{\sc next}(W, \Attr, head(\gamma))=\mbox{\sc impossible}$,   and  $F'$ outputs 
$(W,Attr',\Tail(\gamma))$, whenever
\LIB{$}$\mbox{\sc next}(W, \Attr, head(\gamma))=Attr'$\LIB{$} for some attribute description $Attr'$. 
Since {\sc next} runs in deterministic logspace, also $F'$ and $F$ do, and therefore 
$f$ is a logspace computable function.

By construction and by Lemma~\ref{lem:next}, {\sc pathnode} precisely computes the attributes $\Attr(\alpha)$ if there is a node $\alpha$ with $\Label(\alpha)=\pi$ in $T(\GG,\HH)$, whereas otherwise 
{\sc pathnode} outputs {\sc wrongpath}.  
Since the function $\ell$ (expressing the length ${\ell(\pi)}$) is clearly in $Q_{\log}$, 
and since \UNLIB{for each $i$,  $f^{i}(V, attr(\alpha_0), \pi)$}\LIB{$f^{i}(V, attr(\alpha_0), \pi)$, for each $i$,}
is of size polynomially bounded in {the input size $|(V, attr(\alpha_0), \pi)|$,}  
$f^{\ell(\pi)}(V, attr(\alpha_0), \pi)$ can be computed by a procedure in $[\![\FDSPACE[\log n]_\pol]\!]^{\log}$,  and therefore, by Lemma~\ref{Lemma1},  in deterministic space 
$O(\log^2 n)$. The same complexity bounds obviously hold for  {\sc pathnode}. 
\end{proof}

By using a procedure {\sc pathnode} according to the above Lemma, we are now ready to formulate an algorithm 
{\sc decompose}
that computes the decomposition tree $T(\GG,\HH)$ 
to a {\DUAL} instance  $(\GG,\HH)$. In particular, the algorithms first lists the vertices and then the edges of the tree 
$T(\GG,\HH)$.
\medskip

\FULL{\begin{quote}}
\noindent{\bf Algorithm} {\sc  decompose:}\\  
\noindent Input: {\DUAL}-instance  $I=(\GG,\HH)$; \ \  Output: $T(\GG,\HH)$.\\
\noindent {\sc begin}\\
\noindent{\sc output}("Vertices:");\\
\noindent {\sc for} each path descriptor $\pi\in\PD(I)$  {\sc do}\\
\FULL{{\sc if} {\sc pathnode}$(I,\pi)\neq${\sc wrongpath} {\sc then} {\sc output}(\,{\sc pathnode}$(I,\pi)$\,);\\}
\PODS{{\sc if} {\sc pathnode}$(I,\pi)\neq${\sc wrongpath} \\ 
\phantom{.}\noindent\hspace{0.5cm}
{\sc then} {\sc output}(\,{\sc pathnode}$(I,\pi)$\,);\\}
\noindent{\sc output}("Edges:");\\
\noindent {\sc for} each pair $(\pi,\pi')$ of consecutive path descriptors in $\PD(I)$ {\sc do}\\
\phantom{.}\noindent\hspace{0.5cm} {\sc begin}\\
\phantom{.}\noindent\hspace{0.5cm} $\alpha:=\mbox{\sc pathnode}(I,\pi)$;\\
\phantom{.}\noindent\hspace{0.5cm} $\alpha':=\mbox{\sc pathnode}(I,\pi')$;\\
\phantom{.}\noindent\hspace{0.5cm} {\sc if} $\alpha'\neq${\sc wrongpath} {\sc then} \LIB{\\\phantom{.}\noindent\hspace{1.0cm}}
{\sc output}(\,$ \langle label(\alpha),label(\alpha')\rangle$\,);\\
\phantom{.}\noindent\hspace{0.5cm} {\sc end}\\
 \noindent {\sc end.}
\FULL{\end{quote}}

\begin{theorem}
The Algorithm {\sc  decompose} computes the decomposition tree $T(\GG,\HH)$ to a {\DUAL} instance 
$(\GG,\HH)$ in space $O(\log^2 n)$.
\end{theorem} 
\begin{proof}
The correctness of the algorithm follows from the correctness of {\sc pathnode} as shown in Lemma~\ref{lem:pathnode}. For the space bound, note that each each path descriptor requires only $O(\log^2  |I|)$ $=$ $O(\log^2 n)$ bits, and that we can thus iterate (by re-using work-space) over all path descriptors and pairs of path descriptors in $O(\log^2 n)$ space. Given that, by Lemma~\ref{lem:pathnode}, {\sc pathnode} also runs in $O(\log^2 n)$ space, the
entire {\sc decompose} algorithm needs only $O(\log^2 n)$ space. 
\end{proof}

\vspace{-0.7cm}

\begin{corollary}\phantom{.}\hfill\\
\FULL{\vspace{-0.7cm}}
\PODS{\vspace{-0.4cm}}
\begin{enumerate}
\item Deciding {\DUAL} is in \mbox{\em DSPACE[$\log^2 n$]}.
\item If $tr(\GG)\neq \HH$, then computing a new transversal of $\GG$ w.r.t.  $\HH$ is in 
\mbox{\em FDSPACE[$\log^2 n$]}.
\end{enumerate}
\end{corollary}
\begin{proof}
In both cases, it is possible to first compute the entire decomposition tree $T(\GG,\HH)$ in  FDSPACE[$\log^2 n$], and then (i) 
for problem 1 check by a DLOGSPACE procedure whether all leaves are marked {\DONE}, and (ii) for problem 2, use an
FLOGSPACE procedure to find a node $\alpha$ marked {\FAIL} in $T(\GG,\HH)$ and output its component $t(\alpha)$. 
Let $\circ$ denote the composition operator for complexity classes in the obvious sense. 
Given that \FULL{FDSPACE[$\log^2 n$]$\circ$DLOGSPACE =DSPACE[$\log^2 n$],}\PODS{$$\mbox{FDSPACE[$\log^2 n$]$\circ$DLOGSPACE =DSPACE[$\log^2 n$],}$$}
 and given that, moreover, 
\FULL{FDSPACE[$\log^2 n$]$\circ$FLOGSPACE =FDSPACE[$\log^2 n$],} 
\PODS{$$\mbox{FDSPACE[$\log^2 n$]$\circ$FLOGSPACE =FDSPACE[$\log^2 n$],}$$}
the complexity bounds follow. 
Alternatively, we can solve the problems 1 and 2 directly by respective slight modifications of 
{\sc decompose}. \nop{For 1 the modified procedure  would just suppress all outputs and terminate in an accepting state 
if no node is marked {\FAIL}. For 2, the modified procedure would also  suppress all outputs, except if it encounters the first node
$\alpha$ marked {\FAIL}, for which which it would output $t(\alpha)$ and then stop.}  
\end{proof} 

\smallskip

Note that if $tr(\GG)\neq\HH$, the witness $t(\alpha)$ produced is not necessarily a {\em minimal} transversal of $\GG$, but is, in general, just a transversal of $\GG$ that contains no edge of $\HH$ and thus witnesses that $tr(\GG)\neq\HH$, because $t(\alpha)$ must {\em contain} a missing minimal transversal of $\GG$. From $t(\alpha)$, such a minimal transversal $t$ can easily be computed in  polynomial time by letting first $t:=t(\alpha)$ and by then successively eliminating vertices $v$  from $t$ for which $t-\{v\}$ is still a transversal of $\GG$. However this process requires linear space  in the vertex set $V$ to remember the eliminated vertices plus logarithmic space in the instance size $|(\GG,\HH)|$ for checking. This is still better than polynomial space in the full instance size, but is not quite in quadratic logspace. It is currently not clear whether there exists a smarter algorithm that requires quadratic logspace only. 

\FULL{\section{Tightening the  Complexity Bound}}
\PODS{\section{{A Tighter Bound for {\sc\bf{{\LARGE D}UAL\ }}}}}
By the results of the previous section, {\DUAL} and its complement $\coDUAL$ are in quadratic logspace, i.e., in \LIB{the class}  DSPACE[$\log^2 n$]. 
On the other hand, as already mentioned, the complement of  
{\DUAL} is in $\beta_2$P,  the class of problems solvable in polynomial time with 
$O(\log^2 n)$ nondeterministic guesses. $\beta_2$P\UNLIB{This class} is identical with the complexity class $\GC(\log^2 n,{\rm PTIME})$ of the so called {\em Guess and Check} model for limited nondeterminism~\cite{CaiChen93,Goldsmith}, where $O(\log^2 n)$ nondeterministic bits are guessed  and are appended to the input before the proper PTIME computation starts. The Guess and Check classes are, more generally, defined as follows. Let  $C$ be a complexity class and $s$ a numerical function.  Then $\GC(s(n),C)$ is the class 
of all languages $L$ for which there exists a language $A\in C$ such that an input string $I$ is in $L$ iff there is a string $J$ of $O(s(|I|))$ bits, such that $(I,J)$ is in $A$. In other words, $L$ is in $\GC(s(n),C)$ iff the membership of a string  $I$ in $L$ can be checked in $C$ after having guessed  $O(s(n))$ nondeterministic bits that can be used as an additional input.

Given that PTIME  is believed to be  incomparable with DSPACE[$\log^2 n$] (cf~\cite{ john-90}), 
and given that, obviously, PTIME $\subseteq\GC(\log^2 n,{\rm PTIME})$, it is very likely that also 
 $\GC(\log^2 n,{\rm PTIME})$ and DSPACE[$\log^2 n$] are incomparable. Since {\DUAL} belongs to both classes, this suggests that neither well  characterizes  {\DUAL}, and that {\DUAL} is unlikely to be complete for either.  This observation incited us to look out for a complexity class containing {\DUAL} that would be contained in both classes $\GC(\log^2(n),{\rm PTIME})$ and DSPACE[$\log^2 n$], that wou\-ld thus constitute a 
tighter upper complexity bound for {\DUAL} than all those we have seen so far. In this section, we present precisely such a complexity class.  In order to describe this class, we state some definitions and prove a lemma.

\smallskip

The class  $\SUPER$ is defined as the composition of  \UNLIB{\PODS{the class \\}} $[\![\FDSPACE[\log n]_\pol]\!]^{\log}$ with LOGSPACE. \UNLIB{(where LOGSPACE stands for \DSPACE[$\log n$])} Formally, $\SUPER$ is equal to
$$[\![\FDSPACE[\log n]_\pol]\!]^{\log}\circ \mbox{\rm LOGSPACE}.$$
Here an input $I$ is first transformed to an output $O$
by a \FULL{$[\![\FDSPACE[\log n]_\pol]\!]^{\log}$ procedure,} 
\PODS{functional procedure that runs in deterministic space $[\![\FDSPACE[\log n]_\pol]\!]^{\log}$,}
after which $O$ is submitted to 
a LOGSPACE decision procedure which will decide based on $O$ if the original input $I$ 
is accepted or rejected.

{Note that $\SUPER$ is by all means a complexity class defined in terms of machines and resource bounds. In addition to the classical resources such as the amount of workspace, we here involve somewhat more unusual resources such as the allowed number of self-compositions, which is here bounded by $O(\log n)$, whence the superscript $\log$,  and the allowed size of intermediate outputs in compositions, which is here polynomially bounded, whence the subscript $\pol$.}    

\begin{lemma}
\label{lem:failpath}
Given a {\DUAL} instance $I=(\GG,\HH)$ and a path descriptor $\pi\in \mbox{\it PD}(I)$, deciding \UNLIB{whether}\LIB{if} {\sc pathnode}($I,\pi$) outputs a leaf of $T(\GG,\HH)$ whose {mark}-component is {\FAIL} \FULL{is in} \PODS{is feasible \UNLIB{within}\LIB{with complexity}} {\em \SUPER}.
\end{lemma}
\begin{proof} 
The proof of Lemma~\ref{lem:pathnode} already shows that {\sc path\-node} \UNLIB{is}\LIB{lies in} 
 $[\![\FDSPACE[\log n]_\pol]\!]^{\log}$. Deciding whether {\sc pathnode}($I,\pi$) outputs a leaf of $T(\GG,\HH)$ whose {\em mark}-com\-po\-nent is {\FAIL} can thus be implemented by first executing {\sc pathnode}($I,\pi$), and then checking whether the output is a node labeled {\FAIL}. This is obviously in \LIB{}$[\![\FDSPACE[\log n]_\pol]\!]^{\log}\circ \mbox{\rm LOGSPACE}\PODS{\!}=\PODS{\!} \SUPER$\LIB{}\end{proof}

\FULL{Finally, we study the main class of this section:  $\GC(\log^2 n, \SUPER)$.}
\PODS{\LIB{We next}\UNLIB{Let us finally} consider  $\GC(\log^2 n, \SUPER)$, the main complexity class studied in this section.}
\begin{theorem} 
\label{theo:dualinclass}
{\em $\coDUAL\in \GC(\log^2 n, \SUPER)$}.
\end{theorem}

\begin{proof}
In order to find a new transversal $t$ of $\GG$ with respect to $\HH$ for a {\DUAL} instance 
$I=(\GG,\HH)$, rather than computing the entire 
decomposition tree $T(\GG,\HH)$,  it is sufficient to  guess a branch of this tree that terminates in a leaf $\alpha$ labeled {\FAIL}, and then compute $t(\alpha)$. Guessing such a branch amounts to guess a  
path descriptor $\pi$ and then checking that {\sc pathnode($I,\pi$)} outputs a node marked {\sc fail}.
Guessing $\pi$ amounts to guess $\log^2 n$ bits, and this is all our guess-and-check algorithm guesses. Checking that $\pi$ is a {\FAIL} node is, by 
Lemma~\ref{lem:failpath} in  $\SUPER$, 
hence the overall computation \UNLIB{can be done within}\LIB{is in}  $\GC(\log^2 n, \SUPER)$.
\end{proof}

The last theorem of this section shows, as promised,  that $\GC(\log^2 n, \SUPER)$ is effectively a subclass of the tightest known \LIB{other} upper bounds that are most likely incomparable to each other:
 \LIB{the classes} DSPACE[${\log^2 n}$] and $\GC(\log^2n,{\rm P})=\beta_2{\rm P}$.

\begin{theorem}
\label{theo:intersection}
{\em $\GC(\log^2 n, \SUPER)\FULL{\subseteq}$}\, \PODS{$\subseteq$} 
{\em   DSPACE [${\log^2 n}${\rm ]}$\,\cap$\,$\GC(\log^2 n,{\rm PTIME})$}.
\end{theorem}   
\begin{proof}\LIB{\hspace{-0.0mm}}For the inclusion $$\GC(\log^2 n, \SUPER)\, \subseteq\, \mbox{\rm DSPACE [}{\log^2 n}\mbox{\rm ]},$$ note that a decision procedure  in \PODS{the complexity class} $\GC(\log^2 n, \SUPER)$
amounts to (i) guessing $O(\log^2 n)$ bits, which  can be simulated by an exhaustive enumeration of all possible guesses (under re-use of space), which is feasible in 
{\rm DSPACE [}${\log^2 n}${\rm]}, and (ii)  for each such simulated guess, performing a check \FULL{that lies} in \FULL{the complexity class}  
$[\![\FDSPACE[\log n]_\pol]\!]^{\log}\circ \mbox{\rm LOGSPACE}$. 
Since, by Lemma~1, $[\![\FDSPACE[\log n]_\pol]\!]^{\log}\subseteq \mbox{\rm DSPACE [}{\log^2 n}{\rm ]}$, and given that the composition of a function from the class
$\mbox{\rm FDSPACE [}{\log^2 n}\mbox{\rm ]}$
with a LOGSPACE computation yields a 
$\mbox{\rm DSPACE [}{\log^2 n}\mbox{\rm ]}$
decision procedure, the overall computation is in 
$\mbox{\rm DSPACE [}{\log^2 n}\mbox{\rm ]}$.

To establish the reverse inclusion  
  $$\mbox{$\GC(\log^2 n, \SUPER)\LIB{$ $}\subseteq \GC(\log^2 n,{\rm PTIME})$},$$ it 
\UNLIB{is sufficient to see that}\LIB{is obviously sufficient  to see that} 
 $\SUPER$ is contained in {\sc PTIME}. In fact, a decision procedure in {\SUPER} 
amounts to a pipelined execution of  $O(\log n)$ instantiations of a  logspace function $f$, where the intermediate results are guaranteed to be of polynomial size in the original input, followed by the application of a logspace Boolean procedure $g$. This can   be replaced by the pipelined  execution of $O(\log n)$ instances of  a PTIME procedure equivalent to $f$, followed by the application of a Boolean PTIME procedure equivalent to $g$. In total, this latter process is in PTIME because it amounts to a logarithmic number of invocations of a PTIME procedure, where each time the input size is bounded by a polynomial in the size $n$ of the overall input. Therefore,  $\SUPER\subseteq \mbox{\sc PTIME}$.
\end{proof}

\FULL{\section{Summary, Discussion, and Conclusion}}
\PODS{\section{{Discussion and Conclusion}}}
In this paper we have derived new complexity bounds for the \DUAL problem and its complement  {$\coDUAL$}, that show that these problems can, in principle, be implemented by space-efficient algorithms. These bounds are depicted in Figure~1 in relation to the other relevant complexity classes. Here, set-inclusion is  
visualized by ascending lines or paths. We believe that our results represent some progress in the long and rather tortuous battle towards a better understanding of the mysterious {\DUAL} problem. Our results are ---for the time being---  mainly of  theoretical interest, and we do not claim they have immediate practical consequences. In fact, it is currently not clear whether our space-efficient version of the algorithm by Boros and Makino has any {\em practical}

\LIB{

\phantom{supercalifragilistic}

\vspace{0.33cm}

\begin{figure}[\LIB{t}\UNLIB{h}]
\begin{center}
\includegraphics[trim = 0.5cm -1.0cm 1.0cm 1.0cm,  
height=8.5cm,width=0.58\textwidth]{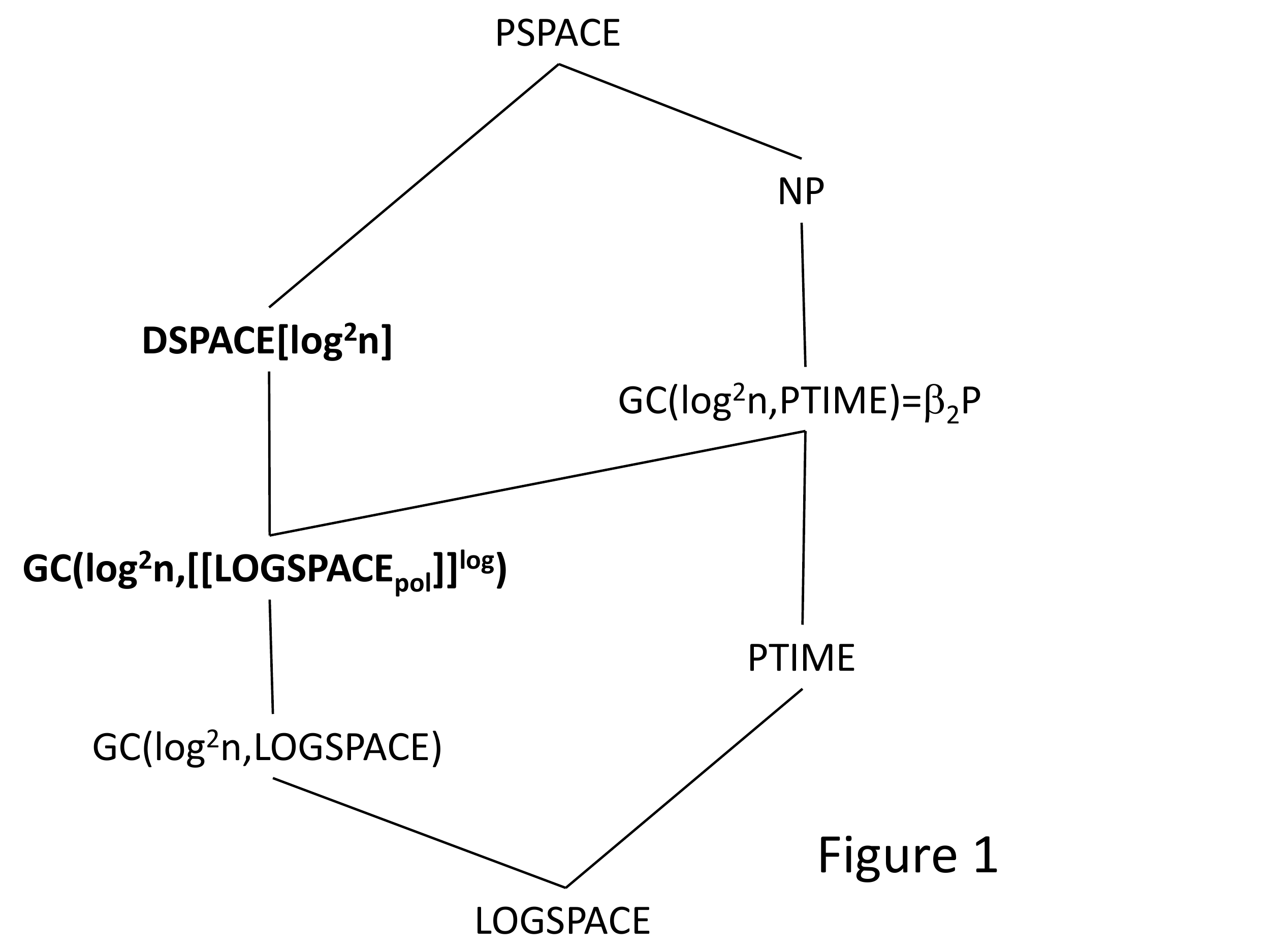}
\end{center}
\end{figure}
}
\PODS{\noindent}
\LIB{\hspace{-1mm}}
\FULL{
\begin{wrapfigure}{r}{0.6\textwidth}
\begin{center}
\includegraphics[trim = 0.0cm -1cm 0.9cm 1.0cm,  height=7.0cm,width=0.7\textwidth]{Figure1}
\end{center}
\end{wrapfigure}
}
{\noindent}\hspace{-1.5mm} advantage over its original version. Future research may look at the applicability of such space-efficient techniques in presence of extremely large hypergraphs or data relations. Our bounds do prove useful for other purposes, however.

For instance, the $O(\log^2 n)$  space bound helps telling the \LIB{}{\DUAL} problem apart from other problems that are candidates for completeness for intermediate classed between P and NP. For example, model-checking modal $\mu$-calculus formulas~\cite{EJS93}, and the equivalent problem of whether a given player has a winning strategy in a parity games on graphs~\cite{EJ91,ZIE98}, are such problems. They are not known to be tractable but lie 
in  UP\,$\cap$\,coUP~\cite{JUR98}, and are thus most likely not NP complete either. Given that 
these problems bear a certain superficial similarity to {\DUAL}, the question arises, whether they are actually disguised versions of sub-problems of  {\DUAL}, 
and can thus be reduced to {\DUAL} via simple low-level reductions (logspace or lower). By our results, this turns out not to be the case unless PTIME is in DSPACE[$\log^2 n$], which is highly unlikely. In fact,  the model-checking problem for the modal  
$\mu$-calculus and the equivalent parity game problem are known to be PTIME hard (see (\cite{HENZI03}). 


\nop{Firstly, the $O(\log^2 n)$  space bound indicates that there exist space-efficient algorithms, and this encourages us to look for {\em practical} space efficient
 solution methods  for {\DUAL} and its equivalent problems in data mining and in the database area.  We feel that space-efficiency may be an advantage, when dealing with big data stemming from financial transactions or biological experiments and so on. When mining terabytes of data, we might want to trade workspace (i.e., main memory) for runtime. Our results state that this is, in principle, feasible. \nop{Without any pretension, let us just dare a {\em Gedankenexperiment}. Assume we have a  terabyte of financial data to mine, e.g., in form of a  hypergraph whose vertices are a very large set of possible 
financial transactions and whose edges are sets of transactions that happen in the same millisecond. Since $\log^2(1TB)< 2MB$, at least the space requirement for the mining tasks are brought down to earth. } Future research will show whether our space bound can be exploited to come up with a reasonable algorithm.}

 We hope, moreover, that our results may help steering future  research towards a matching bound for the {\DUAL} problem. We have reasons not to believe that $\coDUAL$ is hard for \LIB{the class} $\GC(\log^2 n, \SUPER)$. This upper bound, however, gives us some intuition of where to dig further. For example, we conjecture that $\coDUAL$ lies in 
$\GC(\log^2 n, \mbox{\rm LOGSPACE})$, and hope to be able to prove this in the near future. 

Other future work will include the further analysis of hypergraph decomposition techniques to deal with the {\DUAL} problem. It is known that  {\DUAL} is tractable for hypergraphs of bounded degeneracy, and, in particular, for  acyclic hypergraphs~\cite{EGM03}. The latter coincide with  all hypergraphs of hypertree width 1 (see~\cite{gottlob-etal-98,GLS-survey,AGG07}). However, it was shown in~\cite{EGJelia} that for hypergraphs whose hypertree width is bounded by some constant $k\geq 2$, the  $\DUAL$ problem remains as hard as in the general case. It would thus be interesting to find hypergraph decomposition methods that are more general than 
bounded degeneracy and still lead to tractable $\DUAL$ instances. Other research directions are to look for new parameters that lead to tractable (or even fixed-parameter tractable) instances in case they are bounded. See, for instance,~\cite{hagenDiss} for some fixed-parameter tractability results, and~\cite{sacca} for a new parameter leading to tractability. 
\PODS{\UNLIB{
\bigskip

\begin{figure}[\LIB{t}\UNLIB{h}]
\begin{center}
\includegraphics[trim = 0.0cm -1cm 0.9cm 1.0cm,  
height=7.2cm,width=0.58\textwidth]{Figure1}
\end{center}
\end{figure}

}}

\section*{Acknowledgments}
This work was supported by 
\UNLIB{the EPSRC grant EP/G055114/1\ Constraint Satisfaction for Configuration: Logical
Fundamentals, Algorithms, and Complexity.} 
\LIB{the EPSRC grant {``Constraint Satisfaction for Configuration: Logical
Fundamentals, Algorithms, and Complexity"} (EP/G055114/1).} 
The author is  grateful to E.~Allender, E.~Boros, K.~Makino, E.~Malizia,  P.~Rossmanith, D. Sacc{\`a} and  H.~Vollmer for help with technical questions and references,  and to the reviewers for useful comments. 

\bibliographystyle{abbrv}
\bibliography{library}
\end{document}